\documentclass[10pt,article,oneside,oldfontcommands]{memoir}

\usepackage{blindtext}
\usepackage{appendix}
\usepackage{hologo}
\usepackage{amsmath}
\usepackage{amsthm}
\usepackage{amssymb}
\usepackage{amsfonts}
\usepackage{url}
\usepackage{eso-pic}
\usepackage{graphicx}
\usepackage{enumerate}
\usepackage{mathrsfs}
\usepackage{tikz}
\usepackage{hyperref}
\usepackage{cite}
\usepackage{bold-extra}
\usepackage[perpage,symbol*]{footmisc}
\usepackage{mathtools}
\usepackage{atbegshi}
\usepackage{titling}
\usepackage{titlesec}
\usepackage{pgfplots}
\usepackage{afterpage}
\usepackage{color}
\usepackage{authblk}
\usepackage{amscd}
\usepackage{mathtools}
\usepackage{kantlipsum}
\usepackage{float}
\usepackage[Symbol]{upgreek}
\usepackage{esint}
\usepackage{relsize}

\settypeblocksize{0.75\stockheight}{0.7\stockwidth}{*}
\setlrmargins{*}{*}{1}
\setulmargins{*}{*}{1.5}
\checkandfixthelayout[nearest]

\theoremstyle{mytheoremstyle}
\newtheorem{theorem}{Theorem}

\theoremstyle{mytheoremstyle}
\newtheorem{lemma}[theorem]{Lemma}

\theoremstyle{mytheoremstyle}
\newtheorem{prop}{Proposition}

\theoremstyle{mytheoremstyle}

\theoremstyle{theorem}

\theoremstyle{definition}
\newtheorem{definition}{Definition}

\theoremstyle{remark}

\theoremstyle{theorem}

\DeclareMathOperator{\id}{id}

\DeclareMathOperator{\dom}{Dom}

\def\bs{\boldsymbol}
\def\d{\mathrm{d}}
\def\i{\mathrm{i}}

\counterwithout{section}{chapter}

\title{Quantum Graphs:\\  $ \mathcal{PT}$-symmetry and reflection symmetry of the spectrum}
\author[1]{P. Kurasov}
\author[1,2,3]{B. Majidzadeh Garjani}
\affil[1]{Department of Mathematics, Stockholm University, Sweden}
\affil[2]{Department of Physics, Stockholm University, SE-106 91 Stockholm, Sweden}
\affil[3]{Nordita, Royal Institute of Technology and Stockholm University, Roslagstullsbacken 23, SE-10691 Stockholm, Sweden}

\newcommand\raisepunct[1]{\mathpunct{\raisebox{0.5ex}{#1}}}

\begin{document}
\maketitle

\thispagestyle{empty}

\begin{abstract}
Not necessarily self-adjoint quantum graphs -- differential operators on metric graphs -- are considered. Assume in addition that the underlying metric 
graph possesses an automorphism (symmetry) $ \mathcal P $. If the differential operator is $ \mathcal P \mathcal T$-symmetric, then its spectrum has reflection symmetry
with respect to the real line. Our goal is to understand whether the opposite statement holds, namely whether the reflection symmetry of the spectrum of
a quantum graph implies that the underlying metric graph possesses a non-trivial automorphism and the differential operator is $ \mathcal P \mathcal T$-symmetric.
We give partial answer to this question by considering equilateral star-graphs. The corresponding Laplace operator with Robin vertex conditions possesses reflection-symmetric spectrum if and only if the operator is $ \mathcal P \mathcal T$-symmetric with $ \mathcal P $ being an automorphism of the metric graph.%

\end{abstract}

\section{Introduction}
Our paper is devoted to spectral theory of quantum graphs -- ordinary differential operators on metric graphs. This is one of the most rapidly growing areas of modern mathematical physics
due to its important applications in physics and applied sciences as well as interesting mathematical problems that emerge \cite{kuchment,gutkin,wires,kurasovschrodinger,kurasovBirk}.
It appeared that symmetries of the underlying metric graphs play a very important role in constructing counterexamples and proving spectral estimates
\cite{parzan,boman,pavel2015,kurasov2014rayleigh,post2012}. If the theory of self-adjoint operators on metric graphs is
rather well-understood, the corresponding theory of non-self-adjoint operators is in its incubatory stage \cite{hussein}. The main subject of our studies is precisely non-self-adjoint quantum graphs.

The main goal of our current paper is to understand connections to the
theory of $ \mathcal P \mathcal T$-symmetric operators -- yet another area of mathematical physics that has got a lot of attention recently
\cite{MR1627442,MR1686605,MR3035410,MR1958251,MR3485039}.
 Standard quantum mechanics in one dimension is described
by self-adjoint differential operators leading to purely real spectrum. Extending the set of allowed operators by including  $ \mathcal P \mathcal T$-symmetric ones leads to the spectrum with
reflection symmetry with respect to the real axis, not only as a set but also including multiplicities. (In the classical studies $ \mathcal P $ is the reflection operator $ (\mathcal P f)(x) = f(-x) $ and
$ \mathcal T $ is the time-reversal operator of complex conjugation $ (\mathcal T f) (x) = \overline{f(x)}.$) If a metric graph possesses a certain automorphism 
(symmetry) $ \mathcal P $, then the corresponding differential
operator can be chosen to have  $ \mathcal P \mathcal T$-symmetry, leading to reflection-symmetric spectrum. We would like to understand whether this mechanism is unavoidable for a
quantum graph to have a reflection-symmetric spectrum. 

We consider the case of an equilateral star-graph $\Gamma$, as the one in Figure~\ref{fig:stargraph},  formed by  $ N $ identical edges joined  at the central vertex, together with the Laplace operator acting
on it. The metric graph $ \Gamma $ has a rich symmetry group generated by the permutations of the edges.
If, the so-called standard vertex conditions (continuity of the function and vanishing of the sum of normal derivatives), are introduced, then the corresponding operator is self-adjoint
and the spectrum is real (an infinite set of discrete eigenvalues tending to $+ \infty$). One may break the self-adjointness by introducing Robin conditions with non-real parameters at the degree-one
vertices. It is relatively easy to see that if the set of Robin parameters is invariant under conjugation, then the corresponding Laplace operator is  $ \mathcal P \mathcal T$-symmetric with respect
to a certain automorphism
 $ \mathcal P $ of the underlying metric graph $ \Gamma. $ Then the spectrum possesses reflection symmetry with respect to the real axis. Our main question is whether the opposite
statement holds, namely, whether the reflection symmetry of the spectrum implies  $ \mathcal P \mathcal T$-symmetry of the operator with respect to a certain 
automorphism $ \mathcal P $ 
of the metric graph $ \Gamma $. For an operator with discrete spectrum and eigenfunctions of the operator and its adjoint building a biorthogonal basis, one may easily construct a symmetry
operator in the Hilbert space, provided the spectrum is reflection symmetric. Such symmetry operator is defined using the eigenfunctions and there is no guarantee that it comes from an
automorphism of the metric graph.
Our main goal is to prove the following theorem:

\begin{theorem}[Main Theorem]\label{thm:maintheorm}
Let  $ L_{\bs{h}} $ be the Laplace operator acting on the equilateral star-graph $ \Gamma $ with the domain given by standard vertex conditions at the internal vertex and
Robin conditions including complex valued parameters $ h_i $  at the degree-one vertices. For the spectrum of the operator $  L_{\bs{h}} $ to possess reflection symmetry
with respect to the real line it is necessary and sufficient that the quantum graph $ L_{\bs{h}} $ is $ \mathcal{PT}$-symmetric, where $ \mathcal P$ is a certain 
automorphism of the
underlying metric star-graph $ \Gamma.$
\end{theorem}

We doubt that the theorem holds for arbitrary metric graphs, but it might be interesting to characterize the class of graphs and vertex conditions for which a similar
statement holds at least for the Laplace operator. In the case under consideration the metric graph $ \Gamma $ is an equilateral star-graph and the group of its automorphisms is
clear. Moreover, we do not consider all possible differential operators on $ \Gamma $ but just the family $ L_{\bs{h}} $ containing a finite
number of (Robin) parameters. Therefore, we are going to prove first a slightly different theorem and obtain the Main Theorem \ref{thm:maintheorm} as
its corollary:

\begin{theorem}\label{thm:2}
Let  $ L_{\bs{h}} $ be the Laplace operator acting on the equilateral star-graph $ \Gamma $ with the domain given by standard vertex conditions at the internal vertex and
Robin conditions including complex valued parameters $ h_i $  at the degree-one vertices. For the spectrum of the operator $  L_{\bs{h}} $ to possess reflection symmetry
with respect to the real line it is necessary and sufficient that the  set of Robin parameters $ \{ h_i \}_{i=1}^N $
 is invariant under conjugation.
\end{theorem}

Spectral theory of non-self-adjoint operators on the star-graph has already been discussed by one of the co-authors\cite{astudillo}. The main difference to the current paper is that
Neumann conditions were assumed at the degree-one vertices, while one considered rotation-invariant conditions at the central vertex. Operators with the Robin
conditions at degree-one vertices were considered by M.\.Znojil\cite{znojil1,znojil2}. All these three papers are devoted to the
 problem of proving that the spectrum of the operator
possesses reflection symmetry (with respect to the real or imaginary axis), provided the vertex conditions have a certain special form. Our interest is in proving the opposite
statement.

The paper is organized as follows. In Section \ref{sec:statement} we describe the model and discuss its elementary spectral properties.
Section \ref{sec:sufficientcondition} is devoted to an easier direction of Theorem \ref{thm:2}, proving if the set of Robin parameters $ \{ h_i \}_{i=1}^N $ 
is invariant under conjugation, then the spectrum has reflection symmetry with respect to the real axis. 
The opposite (harder) statement is proven in the Section \ref{sec:necessarycondition}. 
The last section is devoted to the proof of the Main Theorem \ref{thm:maintheorm}. The proofs of elementary Propositions \ref{prop:symmetry} and \ref{prop:symmetricpolynomial}
can be found  in the Appendix.

\section{Description of the Model and Definitions}\label{sec:statement}
In this section we define rigorously the Laplace operator on the equilateral star-graph $\Gamma$, depicted in Figure~\ref{fig:stargraph}, with complex Robin parameters at the degree one vertices.
The operator we define should be consistent with the geometric picture -- the vertex conditions should connect together values of the functions at each vertex separately.
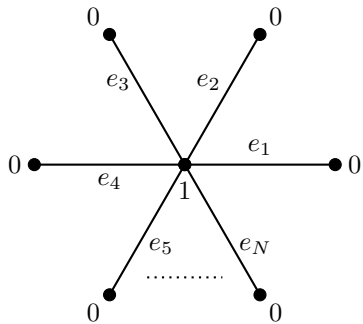
\begin{figure}[H]
\begin{center}
\begin{tikzpicture}
	\draw [thick] (-2,0)--(2,0);
	\draw [thick] (1,-1.73205)--(-1,1.73205);
	\draw [thick] (-1,-1.73205)--(1,1.73205);
	\draw [fill] (1,-1.73205) circle [radius=0.08];
	\draw [fill] (1,1.73205) circle [radius=0.08];
	\draw [fill] (0,0) circle [radius=0.08];
	\draw [fill] (2,0) circle [radius=0.08];
	\draw [fill] (-1,1.73205) circle [radius=0.08];
	\draw [fill] (-1,-1.73205) circle [radius=0.08];
	\draw [fill] (-2,0) circle [radius=0.08];
	\draw [fill] (0,0) circle [radius=0.08];
	\node [below] at (0,-0.1) {$1$};
	\node [right] at (2.05,0) {$0$};
	\node [left] at (-2.05,0) {$0$};
	\node [above left] at (-1,1.73205) {$0$};
	\node [below right] at (1,-1.73205) {$0$};
	\node [above right] at (1,1.73205) {$0$};
	\node [below left] at (-1,-1.73205) {$0$};
	\node [above] at (1,0) {$e_1$};
	\node [below] at (-1,0) {$e_4$};
	\node [above left] at (-0.6,1.73205/2) {$e_3$};
	\node [below right] at (-.6,-1.73205/2) {$e_5$};
	\node [above left] at (.6,1.73205/2) {$e_2$};
	\node [below right] at (0.6,-1.73205/2) {$e_N$};
	\draw [thick, dotted] (-.5,-1.5)--(.5,-1.5);
\end{tikzpicture}
\end{center}\caption{Star-graph $\Gamma$ with $N$ edges.}\label{fig:stargraph} 
\end{figure}
More explicitly let us assume that the graph $ \Gamma$ is built of $ N $ edges $e_1$ till $e_N$, each of unit length. We identify each edge with a  separate copy 
of the interval $ [0,1] $,
where $ 1 $ corresponds to the unique internal vertex.
Consider the Hilbert space  $\mathscr{H}=\text{L}^2\big((0,1), \mathbb C^N \big) \ni \bs{u}=(u_1,\ldots,u_N) $
 equipped with the inner product:
\begin{align}
	\langle\,\bs{u}\,,\bs{v}\,\rangle:=\sum_{i=1}^{N}\int_{0}^{1}\overline{u}_i(x)v_i(x)\,\d x.
\end{align}
Let $h_1$ till $h_N$, known as \emph{Robin parameters}, be given complex numbers and also let $\bs{h}:=(h_1,\ldots,h_N)$.  Consider the \emph{Laplace} operator $L_{\bs{h}}=-\d^2/\d x^2$ with the domain $\dom(L_{\bs{h}})$, consisting of all  complex-valued functions in the Sobolev space $W_2^2 \big( (0,1), \mathbb C^N \big) $ fulfilling the following vertex conditions:
\begin{itemize}
	\item \emph{standard} vertex conditions at the internal vertex:
	\begin{align}
	u_1(1)=u_2(1)=\cdots=u_N(1),\label{eqn:continuity}\\
	u'_1(1)+u'_2(1)+\cdots+ u'_N(1)=0;\label{eqn:normalderivative}
\end{align}
    \item \emph{complex Robin} conditions at the external vertices:
    \begin{align}\label{eqn:robin}
	u'_j(0)=h_j\,u_j(0),\qquad 1\leqslant j\leqslant N.
\end{align}
\end{itemize}

The reader should note that the operator $L_{\bs{h}}$ is not necessarily self-adjoint. This stems from the fact that complex values for Robin parameters are allowed. To see this, let $\bs{u}$ and $\bs{v}$ be two elements of $\dom(L_{\bs{h}})$.  Using integration by parts, we have:
\begin{align*}
	\langle\,L_{\bs{h}}\bs{u}\,,\bs{v}\,\rangle-\langle\,\bs{u}\,,L_{\bs{h}}\bs{v}\,\rangle&=\sum_{i=1}^{N}\bigg(-\int_{0}^{1}\overline{u}''_i(x)v_i(x)\,\d x+\int_{0}^{1}\overline{u}_i(x)v''_i(x)\,\d x\bigg)\\&=
	\sum_{i=1}^{N}\Big(-\overline{u}'_i(x)v_i(x)\big|_{0}^{1}+\overline{u}'_i(x)v_i(x)\big|_{0}^{1}\Big)\\&=
	\sum_{i=1}^{N}(\overline{h}_i-h_i)\overline{u}(0)v_i(0)+\sum_{i=1}^{N}\big[\overline{u}_i(1)v'_i(1)-\overline{u}'_i(1)v_i(1)\big], 
\end{align*}
where Equations~\eqref{eqn:robin} are used. From Equations~\eqref{eqn:continuity} and \eqref{eqn:normalderivative}, one can see that the second sum above vanishes and we get:
\begin{align}
	\langle\,L_{\bs{h}}\bs{u}\,,\bs{v}\,\rangle-\langle\,\bs{u}\,,L_{\bs{h}}\bs{v}\,\rangle=\sum_{i=1}^{N}(\overline{h}_i-h_i)\overline{u}_i(0)v_i(0),
\end{align}
which does not in general vanish, since we allow for complex Robin parameters. Therefore, the spectrum of $L_{\bs{h}}$ might contain complex values.

In what follows, we use a boldface letter to denote an $N$-tuple and  the corresponding Roman letter with subscripts, running from $1$ till $N$, to denote its components. For instance, $\bs{c}$ is used to denote the $N$-tuple $(c_1,\ldots,c_N)$. In this context, we use  $\overline{\bs{c}}$ to indicate $(\overline{c}_1,\ldots,\overline{c}_N)$ and use  $\bs{c}_i$, $1\leqslant i\leqslant N$, to indicate the $(N-1)$-tuple obtained from $\bs{c}$ by suppressing its $i$th component. \noindent For example,  if $\bs{c}=(-\i,1,3\i,-3\i,\i)$, then $\overline{\bs{c}}=(\i,1,-3\i,3\i,-\i)$ and $\bs{c}_3=(-\i,1,-3\i,\i)$.
\begin{definition}
	An $N$-tuple $\bs{c}$ is called \emph{invariant under conjugation}, if non-real components of this tuple can be grouped in conjugate pairs.
\end{definition}
\noindent For example, $\bs{c}$ introduced above is invariant under conjugation but $\bs{d}=(1,\i,\i,-\i)$ is not.

\begin{definition}\label{dfn:time-reversal}
	The \emph{time-reversal} operator $\mathcal{T}$ is an operator defined on  $\mathscr{H}$ by
\begin{align}\label{eqn:Treversal}
	\mathcal{T}\bs{u}=\overline{\bs{u}}.
\end{align} 
This is an anti-linear operator.
\end{definition}
\begin{definition}\label{dfn:parity}
    Let $i$ and $j$ be integers and $1\leqslant i\leqslant j\leqslant N$. A linear operator  $\mathcal{P}_{i,j}$ acting on $\mathscr{H}$ as follows
\begin{align}
	(\mathcal{P}_{i,j}\bs{u})_k=
	\left\{\hspace{-.2cm}
	\begin{array}{rl}
		u_k,&k\not= i, j\\
		u_j,&k=i\\
		u_i,&k=j
	\end{array}
    \right.,  \; \; 1\leqslant k\leqslant N
\end{align}
is called a \emph{permutation} operator. Note that, for any $1\leqslant i\leqslant N$, $\mathcal{P}_{i,i}$ is the identity operator $\id_{\mathscr{H}}$.
\end{definition}
\noindent Permutations $\mathcal{P}_{i,j}$ are elements of the automorphism group for the metric star-graph $\Gamma$. This group naturally induces
 a group of unitary transformations (symmetries) in the Hilbert space $\mathscr{H}$:
\begin{align}
	\mathcal{S}f(x)=f(\mathcal{S}^{-1}x)
\end{align}
for any $ \mathcal S $ from the graph automorphism group. 
If all $h_i$'s are equal then the quantum graph, namely, the operator $L_{\bs{h}}$ possesses the same symmetry group as $ \Gamma$. This is not the case if $ h_i $'s are different.
 
\begin{definition}\label{dfn:S-symmetric}
	A unitary operator $\mathcal{S}$ defined on $\mathscr{H}$ is called a \emph{symmetry} of the quantum graph $L_{\bs{h}}$, if $\mathcal{S}L_{\bs{h}}=L_{\bs{h}}\mathcal{S}$. In this case, the quantum graph  is called $\mathcal{S}$-\emph{symmetric}. For $\mathcal{S}$ to be a symmetry of $L_{\bs{h}}$, it is necessary that $\mathcal{S}\big(\dom(L_{\bs{h}})\big)\subseteq\dom(L_{\bs{h}})$. 
\end{definition}

\section{Sufficient Condition for Theorem \ref{thm:2}}\label{sec:sufficientcondition}

In this section we prove that the invariance of the $N$-tuple  $\bs{h}=(h_1,\ldots,h_N)$ under conjugation implies that the spectrum of $L_{\bs{h}}$ possesses reflection symmetry with respect to the real axis. To 
establish this, one needs the following proposition whose proof is given in Appendix \ref{app:proposition}.

\begin{prop}\label{prop:symmetry}
	Let $\mathcal{AT}$, where $\mathcal{A}$ is an invertible linear operator defined on $ {H}$ and $\mathcal{T}$ is the time-reversal operator, be a symmetry for a linear operator $ L $ acting in a Hilbert space $  H $, namely, $\mathcal{AT}L=L \mathcal{AT}$. Then if $\lambda$ is an eigenvalue of $L $ with degeneracy $d$, then $\overline{\lambda}$ is also an eigenvalue of $L$ with the same degeneracy $d$.
\end{prop}

\begin{lemma}\label{lma:PTsymmetric}
	If the $N$-tuple $\bs{h}=(h_1,\ldots,h_N)$, consisting of Robin parameters, is invariant under conjugation, then there exists an automorphism  $\mathcal{P}$ of the graph $ \Gamma $ and hence an invertible linear operator on $\mathscr{H}$ such that the mentioned quantum graph  $L_{\bs{h}}$ is $\mathcal{PT}$-symmetric, that is, $\mathcal{PT}L_{\bs{h}}=L_{\bs{h}}\mathcal{PT}$.
\end{lemma}

\begin{proof}
	If all components of $\bs{h}$ are real, then it is straightforward to see that, say, $\mathcal{P}_{1,1}\mathcal{T}=\id_{\mathscr{H}}\mathcal{T}$ is a symmetry of the quantum graph. Therefore, suppose that this is not the case and, since $\bs{h}$ is assumed to be invariant under conjugation, suppose that $\bs{h}$ has at least two non-real components.  Let $2m$, for some integer $m$ where $1\leqslant m\leqslant\lfloor N/2\rfloor$,  be the number of non-real components of $\bs{h}$.  Without loss of generality, one can assume that the components of $\bs{h}$ are coming in an order in which $h_{2j}=\overline{h}_{2j-1}$, for all integers $1\leqslant j\leqslant m$, and all other components (if any) are real.
	
	We now show that $\mathcal{PT}$, where  $\mathcal{P}:=\prod_{j=1}^{m}\mathcal{P}_{2j-1,2j}$, is a symmetry of $L_{\bs{h}}$. Expressed in words, $\mathcal{P}$ is the operator that interchanges the $(2j-1)$th and $(2j)$th ($1\leqslant j\leqslant m$) components of the $N$-tuple it is acting on, and leaves the rest (if any) of the components unchanged.
	
	First, we show that $\dom(L_{\bs{h}})$ is invariant under the action of $\mathcal{PT}$. Let $\bs{u}$ be in $\dom(L_{\bs{h}})$ and let $\bs{v}=\mathcal{PT}\bs{u}$. Then, for every integer $j$ in the interval $[1,m]$, $v_{2j-1}=\overline{u}_{2j}$ and $v_{2j}=\overline{u}_{2j-1}$ and, for every integer $j$ in the interval $[2m+1,N]$,  $v_j=\overline{u}_j$. Therefore, Equations~\eqref{eqn:continuity} and \eqref{eqn:normalderivative} are fulfilled for components of $\bs{v}$, since the corresponding equations for $\bs{v}$ components are complex conjugates of, at most, a rearrangement of their $\bs{u}$-component counterparts. To see that Robin vertex conditions are also satisfied for $\bs{v}$ components, we show that $v'_k(0)=h_k\,v_k(0)$, for every integer $k$ in the interval $[1,N]$. 
	
	If $k=2j-1$, for some integer $j$ in the interval $[1,m]$, then $v_k=\overline{u}_{k+1}$ and $h_k=\overline{h}_{k+1}$. Hence $v'_k(0)=h_k\,v_k(0)$ if and only if $\overline{u}'_{k+1}(0)=\overline{h}_{k+1}\,\overline{u}_{k+1}(0)$ or, equivalently,  $u'_{k+1}(0)=h_{k+1}\,u_{k+1}(0)$. For even integers $k$ in this interval, the proof is similar. If $k$ is an integer in $[2m+1,N]$, provided that this interval is non-empty, then $v_k=\overline{u}_k$ and $h_k=\overline{h}_k$. Thus $v'_k(0)=h_k\,v_k(0)$ is equivalent to $u'_k(0)=h_k\,u_k(0)$ in this case. 
	
	Checking that $L_{\bs{h}}\mathcal{PT}\bs{u}=\mathcal{PT}L_{\bs{h}}\bs{u}$, for any $\bs{u}$ in the domain of $L_{\bs{h}}$, is striaghtforward.
\end{proof}

The fact that the invariance of the $N$-tuple of Robin parameters under complex conjugation implies
reflection symmetry of the spectrum, follows simply as a combination of Proposition~\ref{prop:symmetry} with Lemma~\ref{lma:PTsymmetric}.

\section{Necessary Condition  for Theorem \ref{thm:2} }\label{sec:necessarycondition}
In this section we assume that the spectrum of $L_{\bs{h}}$ is invariant under complex conjugations 
 and we prove that $\bs{h}$ is invariant under conjugation. To this end, we start by determining the secular equation  for $L_{\bs{h}}$.  
\subsection{The Secular Equation}

Let $\lambda:=z^2$ be a \emph{non-zero} eigenvalue of $L_{\bs{h}}$.   Hence, there exists a non-zero $\bs{u}$ in $\dom(L_{\bs{h}})$ such that
\begin{align*}
	L_{\bs{h}}\bs{u}=z^2\,\bs{u}.
\end{align*}
In particular, after writing it in component form, we have:
\begin{align}\label{eqn:eigenvalue}
	-\frac{\d^2}{\d x^2}u_i(x)=z^2\,u_i(x),\qquad i=1,2,\ldots,N.
\end{align}
The general solutions of Equations~\eqref{eqn:eigenvalue}, for non-zero $z$, have the following form:
\begin{align}\label{eqn:generalsolution}
	u_i(x)=A_i\,\cos zx+B_i\,\sin zx,\qquad i=1,2,\ldots,N,
\end{align}
where $A_i$'s and $B_i$'s are complex constants. 
Taking into account the vertex conditions \eqref{eqn:robin} on them, we obtain
$ B_i=(h_iA_i)/z$  and
\begin{equation} \label{alpha}
u_i(x)=\frac{A_i}{z}  \underbrace{ \left(z \cos zx+ h_i\,\sin zx \right) }_{\displaystyle := \alpha_{h_i} (z)} . 
\end{equation}
Taking into account the vertex conditions \eqref{eqn:continuity} on these functions, one gets:
\begin{align}\label{eq:008}
	\alpha_{h_1}(z)\,A_1=\alpha_{h_2}(z)\,A_2=\cdots=\alpha_{h_N}(z)\,A_N,
\end{align} 
where $\alpha_{h_i}(z) $ were introduced in (\ref{alpha}).
Vertex condition \eqref{eqn:normalderivative} give rise to the following equation:
\begin{align}\label{eq:010}
	\beta_{h_1}(z)\,A_1+\beta_{h_2}(z)\,A_2+\cdots+\beta_{h_N}(z)\,A_N=0,
\end{align}
where
\begin{align}
	\beta_{h_i}(z):=-z\sin z+h_i\,\cos z,\qquad(i=1,2,\ldots,N).
\end{align}

Equations~\eqref{eq:008} along with Equation~\eqref{eq:010} form a system of $N$ linear equations for $N$ unknowns $A_1$ till $A_N$. For this system to have a non-trivial solution, it is necessary and sufficient that the determinant of the coefficients vanishes, namely,
\begin{align}\label{eqn:secularequations}
	D_{\bs{h}}(z)=0,
\end{align}
where
\begin{align}
    D_{\bs{h}}(z):=\left|
	\begin{matrix}
		\alpha_{h_1}(z)&-\alpha_{h_2}(z)&0&0&\cdots&0&0\\
		0&\alpha_{h_2}(z)&-\alpha_{h_3}(z)&0&\cdots&0&0\\
		0&0&\alpha_{h_3}(z)&-\alpha_{h_4}(z)&\cdots&0&0\\
		0&0&0&\alpha_{h_4}(z)&\ddots&0&0\\
		\vdots&\vdots&\vdots&\vdots&\ddots&\vdots&\vdots\\
		0&0&0&0&\cdots&\alpha_{h_{N-1}}(z)&-\alpha_{h_N}(z)\\
		\beta_{h_1}(z)&\beta_{h_2}(z)&\beta_{h_3}(z)&\beta_{h_4}(z)&\cdots&\beta_{h_{N-1}}(z)&\beta_{h_N}(z)
	\end{matrix}
	\right|.
\end{align}
Equation~\eqref{eqn:secularequations} is the \emph{secular} equation of $L_{\bs{h}}$. Expanding the determinant above along the last row and noting that each of the $N$ resultant determinants is in triangular form, one gets the following relation for $D_{\bs{h}}(z)$:

\begin{align}\label{eqn:D}
	D_{\bs{h}}(z)=\sum_{i=1}^N\Bigg(\beta_{h_i}(z)\prod_{\substack{j=1\\j\not=i}}^{N}\alpha_{h_j}(z)\Bigg).
\end{align}
The function $ D_{\bs{h}}(z) $ is an entire function, since the functions $\alpha$ and $\beta$  are entire. Secondly, 
\begin{align}\label{eqn:Dconjugate}
	\overline{D_{\bs{h}}(\overline{z})}=D_{\overline{\bs{h}}}(z),
\end{align}
since $\overline{\alpha_{h_i}(\overline{z})}=\alpha_{\overline{h}_i}(z)$ and $\overline{\beta_{h_i}(\overline{z})}=\beta_{\overline{h}_i}(z)$.

\subsection{On the Roots of the Secular Equation}
When all the Robin parameters are zero, that is, when $\bs{h}=\bs{0}$, Equation~\eqref{eqn:D} reduces to
\begin{align}\label{eqn:D0}
	D_{\bs{0}}(z)=-N\,z^N&\sin z(\cos z)^{N-1}. 
\end{align}
The non-zero roots of $D_{\bs{0}}(z)$ then are $n\,\uppi$ (for any non-zero integer $n$), each  of multiplicity one, and $n\,\uppi+\uppi/2$ (for any integer $n$), each of multiplicity $N-1$. We concentrate ourselves here on the former set of roots, which are easier to handle.  Now consider the case in which $\bs{h}\not=\bs{0}$, and let $\widetilde{z}_n(\bs{h})$, for some fixed positive integer $n$,  be a root of the secular equation $D_{\bs{h}}(z)=0$. Then we write $\widetilde{z}_n(\bs{h})$  in the following form:
\begin{align}\label{eqn:z_n}
	\widetilde{z}_n(\bs{h})=n\,\uppi+\mathit{\Delta}_n(\bs{h}),
\end{align}
where $\mathit{\Delta}_n(\bs{h})$ denotes the deviation of $\widetilde{z}_n(\bs{h})$, as a root  of equation $D_{\bs{h}}(z)=0$, from $n\,\uppi$. 
 We show that, for sufficiently large $n$,  the correction term $\mathit{\Delta}_n(\bs{h})$ takes the following form:
\begin{align}\label{eqn:delt_n}
	\mathit{\Delta}_n(\bs{h})=  \frac{a_1(\bs{h})}{n}+\frac{a_3(\bs{h})}{n^3}+\frac{a_5(\bs{h})}{n^5}+\cdots\raisepunct{,}
\end{align}
with some coefficients $a_i(\bs{h})$. The general structure of these coefficients  will be determined later.

First we engineer a complex-valued function $K_\epsilon$, in which $\epsilon$ is a positive real parameter,  so that
\begin{align}\label{eqn:implication}
	D_{\bs{h}}(n\,\uppi+z)=0\Leftrightarrow K_{\frac{1}{n\,\uppi}}(z)=0,
\end{align}
for any complex number $z$ and for any positive integer $n$. Employing the elementary formulas:
\begin{align*}
	\sin(n\,\uppi+z)&=(-1)^n\,\sin z,\\
	\cos(n\,\uppi+z)&=(-1)^n\,\cos z,
\end{align*}
it is straightforward to see that 
\begin{align}
	K_{\epsilon}(z):=\sum_{i=1}^{N}\Bigg\{\big[-(1+\epsilon\,z)\sin z+\epsilon h_i\,\cos z\big] \prod_{\substack{j=1\\j\not=i}}^{N}\big[(1+\epsilon\,z)\cos z+\epsilon h_j\,\sin z\big] \Bigg\}\raisepunct{,}
\end{align}
does the job.

We now show that there exists a neighborhood $\mathcal{N}$ of the origin and a positive number $\epsilon_0$ such that for every $\epsilon$ in the interval $(0,\epsilon_0)$, the function $K_\epsilon$ has a unique root $\mathit{\Delta}(\epsilon)$ in $\mathcal{N}$. More importantly, we show that the dependence of $\mathit{\Delta}$ on $\epsilon$ is analytic. Note that $\mathit{\Delta}$ is, in general, a complex-valued function of $\epsilon$.

The equation $K_{\epsilon}(z)=0$ can be considered as a polynomial equation in $\epsilon$:
\begin{align}\label{eqn:g-eqn}
	g_0(z)+g_1(z)\,\epsilon+g_2(z)\,\epsilon^2+\cdots+g_N(z)\,\epsilon^N=0,
\end{align}
in which the coefficient-functions $g_i(z)$ are entire trigonometric functions of $z$. Among all these coefficient-functions, what we actually need is the explicit form of just the first function $g_0(z)$. It is straightforward to see that 
\begin{align*}
	g_0(z)=-N\,\sin z(\cos z)^{N-1}.
\end{align*}
For this function we have
\begin{align*}
	g_0(0)=0,\quad g_0'(0)=-N,
\end{align*}
and, consequently, the entire function $g_0$ is invertible in a neighborhood of the origin, where its first derivative is separated from zero.  Using the inverse of $g_0$, the Equation~\eqref{eqn:g-eqn} can be written in the following manner:
\begin{align}\label{eqn:Gepsilon}
	G_\epsilon(z)=z,
\end{align}
 where
 \begin{align}\label{eqn:Gepsilon-g0}
 	G_\epsilon(z):=-g_0^{-1}\Big(g_1(z)\,\epsilon+g_2(z)\,\epsilon^2+\cdots+g_N(z)\,\epsilon^N \Big).
 \end{align}
Of course, $\epsilon$ in ~\eqref{eqn:Gepsilon-g0} should be chosen sufficiently small.  The reader also notes that for any $z$ in the suitable domain,  $g_0^{-1}(z)=-g_0^{-1}(z)$, since $g_0$ itself is an odd function.
 
 We now have:
\begin{align}\label{eqn:Gprimeepsilon}
 	G_\epsilon'(z)=-\frac{1}{g'_0\Big(g_1(z)\,\epsilon+g_2(z)\,\epsilon^2+\cdots+g_N(z)\,\epsilon^N\Big)}\,\epsilon\big(g'_1(z)+g'_2(z)\,\epsilon+\cdots+g'_N(z)\,\epsilon^{N-1} \big), 
\end{align}
 where prime everywhere denotes the derivative with respect to $z$. As Equation~\eqref{eqn:Gprimeepsilon} shows, $|G_\epsilon'(z)|$ is controlled by $\epsilon$. In particular, there exists  $\epsilon_0>0$ and there exists a neighborhood $\mathcal{N}$ such that  $|G_\epsilon'(z)|<1$, for all $z$ in $\mathcal{N}$ and all $\epsilon$ in the interval $(0,\epsilon_0)$. The fixed-point theorem then implies that the Equation~\eqref{eqn:Gepsilon}, or equivalently, the Equation $K_\epsilon(z)=0$ has a unique root $\mathit{\Delta}(\epsilon)$ in $\mathcal{N}$. To see the analytic dependence of $\mathit{\Delta}$ on $\epsilon$, note that the solution curve $\mathit{\Delta}=\mathit{\Delta}(\epsilon)$ can be viewed as the intersection of the following two analytic manifolds:
 \begin{align}
 	x=G_\epsilon(\mathit{\Delta}),\quad x=\mathit{\Delta},
 \end{align}
 in the three-dimensional space, and therefore it is analytic and one can write:
\begin{align} \label{expD}
	\mathit{\Delta}(\epsilon)=b_0(\bs{h})+b_1(\bs{h})\,\epsilon+b_2(\bs{h})\,\epsilon^2+\cdots\raisepunct{,}
\end{align} 
with some coefficients $b_i(\bs{h})$. The function $ K_\epsilon (z) $ possesses the property $K_{-\epsilon}(-z)=-K_\epsilon(z)$
and, consequently, $\mathit{\Delta} (\epsilon) $ is an odd function of $\epsilon$. Therefore,  all even coefficients in the expansion (\ref{expD}) vanish:
\begin{align} \label{expD2}
	\mathit{\Delta}(\epsilon) = b_1(\bs{h})\,\epsilon+b_3(\bs{h})\,\epsilon^3+\cdots\,\raisepunct{.}
\end{align} 

Now let $n$ be a positive integer such that $1/(n\,\uppi)$ lies in $(0,\epsilon_0)$. Therefore, the unique root $\mathit{\Delta}_{n}(\bs{h})$ of $K_{\frac{1}{n\,\uppi}}$ in $\mathcal{N}$ or, considering \eqref{eqn:implication},  equivalently the unique root of Equation $D_{\bs{h}}(n\,\uppi+z)=0$ in $\mathcal{N}$ is
$\mathit{\Delta}_{n}(\bs{h})=\mathit{\Delta}\big(\frac{1}{n\,\uppi} \big)$, and expansion (\ref{eqn:delt_n}) holds with
\begin{align}
	a_j(\bs{h}):=\frac{1}{\uppi^j}\,b_j(\bs{h}),
\end{align}
for all odd positive integers $j$. 

Our next goal is to discuss formulas for the expansion coefficients, this can be done in terms of
symmetric polynomials in $ h_i $.

\subsection{Elementary Symmetric Polynomials}
In this subsection, we introduce a special kind of symmetric polynomials that allows us to write Equation~\eqref{eqn:D} in a more flexible form, which fits better to our purpose.

\begin{definition}\label{dfn:symmetricpolynomial}
	For a non-negative integer $m$, the $m$th \emph{elementary symmetric polynomial} in $N$ variables $x_1$ till $x_N$ is denoted by $s_m(x_1,\ldots,x_N)$ and it is defined by
\begin{align}\label{eqn:symmetricpolynomial}
	s_m(x_1,\ldots,x_N)=\left\{\hspace{-.2cm}
	\begin{array}{cc}	
	1,&m=0\\[3mm]
	\sum x_{i_1}\!\!\cdots x_{i_m},&1\leqslant m\leqslant N\\[3mm]
	0,&m\geqslant N+1
	\end{array},
    \right. 
\end{align}	
 where the sum above is over all indices $i_1$ till $i_m$ such that $1\leqslant i_1<\cdots<i_m\leqslant N$.
\end{definition}
For instance, the elementary symmetric polynomials in three variables are: 
\begin{align*}
    s_0(x_1,x_2,x_3)&=1,\\
	s_1(x_1,x_2,x_3)&=x_1+x_2+x_3,\\
	s_2(x_1,x_2,x_3)&=x_1x_2+x_2x_3+x_1x_3,\\
	s_3(x_1,x_2,x_3)&=x_1x_2x_3,\\
	s_m(x_1,x_2,x_3)&=0,\qquad(m\geqslant4).
\end{align*}

\noindent It is readily seen that the symmetric polynomial $s_m(x_1,\ldots,x_N)$,  $m\leqslant N$, is  homogeneous of degree $m$ in the sense that for any number $\lambda$,
\begin{align}\label{eqn:homogenity}
	s_m(\lambda\,x_1,\ldots,\lambda\,x_N)=\lambda^m\,s_m(x_1,\ldots,x_N).
\end{align}

In the coming sections, we make use of the properties of elementary symmetric polynomials mentioned in the following proposition, which proof 
can be found in Appendix~\ref{app:proposition}.

\begin{prop}\label{prop:symmetricpolynomial}
	Let $\bs{c}=(c_1,\ldots,c_N)$, and let $i$ and $j$ be integers such that $1\leqslant i\leqslant N$ and   $0\leqslant j\leqslant N-1$. Then
\begin{gather}
    c_i\,s_j(\bs{c}_i)=s_{j+1}(\bs{c})-s_{j+1}(\bs{c}_i)\label{eqn:reccursive},\\
	\sum_{i=1}^{N}s_j(\bs{c}_i)=(N-j)\,s_{j}(\bs{c})\label{eqn:023}.
\end{gather}
\end{prop}
\noindent Moreover, we also make use of the following well-known identity:
\begin{align}\label{eqn:identity}
	\prod_{j=1}^{N}(x+c_j)=\sum_{k=0}^{N}s_k(\bs{c})\,x^{N-k}.
\end{align}

\noindent The following lemma is in some sense the key observation that makes the proof to work.

\begin{lemma}\label{lem:symmetricpolynomials}
An $N$-tuple $\bs{c}=(c_1,\ldots,c_N)$ of complex numbers is invariant under conjugation if and only if, for all $1\leqslant m\leqslant N$, $s_m(\bs{c})$ is real. 
\end{lemma}

\begin{proof}
First, suppose that $\bs{c}$ is invariant under conjugation and let $m$, $1\leqslant m\leqslant N$, be a given integer. From Equation~\eqref{eqn:symmetricpolynomial}, we have: \label{page:A}
\begin{align*}
	\overline{s_m(c_1,\ldots,c_N)}=\sum \overline{c_{i_1}\!\!\cdots c}_{i_m}.
\end{align*} 
Since $\bs{c}$ is invariant under conjugation, for each term in the sum above $\overline{c_{i_1}\!\!\cdots c}_{i_m}\!\!\!\!=c_{j_1}\!\!\cdots c_{j_m}$, where $1\leqslant j_1<\cdots<j_m\leqslant N$. Of course,  not all of these $j$ indices are necessarily distinct from their corresponding $i$ indices. Thus, each term in this sum is either real or its complex conjugate is also a term in this sum and, hence, $s_m(c_1,\ldots,c_N)$ is real.

Conversely, suppose that $s_m(\bs{c})$'s, $1\leqslant m\leqslant N$, are all real. Consider the following polynomial:\label{page:property of A}
\begin{align}
	P(z):=\prod_{i=1}^{N}(z-c_i).
\end{align}
Using Equations~\eqref{eqn:identity} and \eqref{eqn:homogenity}, $P(z)$ can be written as follows:
\begin{align}
	P(z)=z^n+\sum_{m=1}^{N}(-1)^ms_m(\bs{c})\,z^{N-m},
\end{align}
and, therefore, $P(z)$ is a polynomial with real coefficients. Consequently,  a complex number $z_0$ is a root of $P(z)$ if and only if $\overline{z}_0$ is a root and, since the roots of $P$ are $c_1$, $c_2$, \dots, and $c_N$. This would then mean that $\bs{c}$ is invariant under conjugation.
\end{proof}

\subsection{The Secular Equation, New Guise}
We are now ready to write Equation~\eqref{eqn:D} in terms of elementary symmetric polynomials. Plugging $z\cos z$ for $x$ and  $h_j\,\sin z$ for $c_j$ in  Equation~\eqref{eqn:identity}, and using Equation~\eqref{eqn:homogenity}, one can write:
\begin{align*}
	\smash[b]{\prod_{\substack{j=1\\j\not=i}}^{N}}\alpha_{h_j}(z)&=\smash[b]{\prod_{\substack{j=1\\j\not=i}}^{N}}(z\cos z+h_j\,\sin z)\\&\hspace{1.3cm}=
	\smash[b]{\sum_{k=0}^{N-1}}s_k(h_1\,\sin z,\ldots,h_{i-1}\,\sin z,h_{i+1}\,\sin z,\ldots,h_N\,\sin z)\,(z\cos z)^{N-1-k}\\&\hspace{6.8cm}=
	\sum_{k=0}^{N-1}s_k(\bs{h}_i)\,z^{N-1-k}(\sin z)^k(\cos z)^{N-1-k}.
\end{align*}
Hence,
\begin{align*}
	\beta_{h_i}(z)\prod_{\substack{j=1\\j\not=i}}^{N}\alpha_{h_j}(z)&=(-z\sin z+h_i\,\cos z)\sum_{k=0}^{N-1}s_k(\bs{h}_i)\,z^{N-1-k}(\sin z)^k(\cos z)^{N-1-k},
\end{align*}
and after multiplying the terms, we get:
\begin{align}\label{eqn:024}
	\beta_{h_i}(z)\smash[b]{\prod_{\substack{j=1\\j\not=i}}^{N}}\alpha_{h_j}(z)&=-\smash[b]{\sum_{k=0}^{N-1}}s_k(\bs{h}_i)\,z^{N-k}(\sin z)^{k+1}(\cos z)^{N-1-k}\notag\\&\hspace{2cm}+\sum_{k=0}^{N-1}h_is_k(\bs{h}_i)\,z^{N-k-1}(\sin z)^{k}(\cos z)^{N-k}.
\end{align}
The first sum, after the term corresponding to $k=0$ is separated,  can be written in the following form:
\begin{align}\label{eqn:025}
	z^N\sin z(\cos z)^{N-1}+\sum_{k=1}^{N-1}s_k(\bs{h}_i)\,z^{N-k}(\sin z)^{k+1}(\cos z)^{N-1-k}\raisepunct{.}
\end{align}
Using Equation~\eqref{eqn:reccursive}, the second sum in Equation~\eqref{eqn:024} can be rewritten as follows:
\begin{align}\label{eqn:026}
	\smash[b]{\sum_{k=0}^{N-1}}[s_{k+1}&(\bs{h})-s_{k+1}(\bs{h}_i)]\,z^{N-k-1}(\sin z)^{k}(\cos z)^{N-k}\notag\\&\hspace{-.5cm}=
	\smash[b]{\sum_{k=1}^{N}}[s_{k}(\bs{h})-s_{k}(\bs{h}_i)]\,z^{N-k}(\sin z)^{k-1}(\cos z)^{N-k+1}\notag\\&\hspace{.5cm}=
	\sum_{k=1}^{N-1}\big\{[s_{k}(\bs{h})-s_{k}(\bs{h}_i)]\,z^{N-k}(\sin z)^{k-1}(\cos z)^{N-k+1}\big\}+s_N(\bs{h})\,(\sin z)^{N-1}\cos z,
\end{align}
where in the last line, the $N$th term of the sum is written separately, and $s_N(\bs{h}_i)=0$ has been employed. Plugging Equations~\eqref{eqn:025} and \eqref{eqn:026} back into Equation~\eqref{eqn:024}, we obtain:
\begin{multline}\label{eqn:0028}
	\beta_{h_i}(z)\smash[b]{\prod_{\substack{j=1\\j\not=i}}^{N}}\alpha_{h_j}(z)=-z^N\sin z(\cos z)^{N-1}\\\hspace{-1cm}+\smash[b]{\sum_{k=1}^{N-1}}\big\{-s_k(\bs{h}_i)\,z^{N-k}(\sin z)^{k+1}(\cos z)^{N-k-1}\\\hspace{3cm}+[s_{k}(\bs{h})-s_{k}(\bs{h}_i)]\,z^{N-k}(\sin z)^{k-1}(\cos z)^{N-k+1}\big\}\\+s_N(\bs{h})\,(\sin z)^{N-1}\cos z.
\end{multline}
The summand in the middle term above can be written as:
\begin{align*}
	\big\{[s_k(\bs{h})-s_k(\bs{h}_i)]\,\cos^2z-s_k(\bs{h}_i)\,\sin^2z\big\}\,z^{N-k}(\sin z)^{k-1}(\cos z)^{N-k-1},
\end{align*}
or
\begin{align*}
	[s_k(\bs{h})\,\cos^2z-s_k(\bs{h}_i)]\,z^{N-k}(\sin z)^{k-1}(\cos z)^{N-k-1}.
\end{align*}
Plugging back the last expression into Equation~\eqref{eqn:0028}, yields:
\begin{align*}
	\beta_{h_i}(z)&\smash[b]{\prod_{\substack{j=1\\j\not=i}}^{N}}\alpha_{h_j}(z)=-z^N\sin z(\cos z)^{N-1}\\&\hspace{1cm}+\smash[b]{\sum_{k=1}^{N-1}}s_k(\bs{h})\,z^{N-k}(\sin z)^{k-1}(\cos z)^{N-k+1}\\&\hspace{3cm}-
	\smash[b]{\sum_{k=1}^{N-1}}s_k(\bs{h}_i)\,z^{N-k}(\sin z)^{k-1}(\cos z)^{N-k-1}\\&\hspace{7.5cm}
	+s_N(\bs{h})\,(\sin z)^{N-1}\cos z.
\end{align*}
Finally, summing over $i$ and using Equation~\eqref{eqn:023}, one gets:
\begin{align}
	D_{\bs{h}}(z)=-N\,z^N&\sin z(\cos z)^{N-1}\notag\\&+\,\smash[b]{\sum_{k=1}^{N-1}}\big\{s_k(\bs{h})\,(k-N\,\sin^2z)z^{N-k}(\sin z)^{k-1}(\cos z)^{N-k-1}\big\}\notag\\&\hspace{6cm}+N\,s_N(\bs{h})(\sin z)^{N-1}\cos z,
\end{align}
which can be written more compactly as: 
\begin{align}\label{eqn:secularequation}
	D_{\bs{h}}(z)=-N\,z^N&\sin z(\cos z)^{N-1}+z^N\sum_{k=1}^N\frac{s_k(\bs{h})}{z^k}\,f_k(z),
\end{align}
where
\begin{align}
	f_k(z):=(k-N\,\sin^2z)(\sin z)^{k-1}(\cos z)^{N-k-1},\qquad(k=1,2,\ldots,N).
\end{align}

\subsubsection{A Short Detour}\label{subsub:detour} 
Before continuing the main theme, let us have a short detour here to see how practical the formula  given in Equation~\eqref{eqn:secularequation} for the secular equation actually is. Here we use this formula to demonstrate directly that for large $ n$, eigenvalues of $ L_{\bs{h}} $ can be found close to $n\,\uppi$.
 Assuming that $z\not=0$, Equation~\eqref{eqn:secularequation} can be written as follows: 
\begin{align}\label{eqn:perturbed}
	\frac{D_{\bs{h}}(z)}{z^N}=-N\,\sin z(\cos z)^{N-1}+\sum_{k=1}^N\frac{s_k(\bs{h})}{z^k}\,f_k(z).
\end{align}
For any positive integer $n$, $n\,\uppi$ is a root of the first term on the right hand side of Equation~\eqref{eqn:perturbed}. As Equation~\eqref{eqn:D0} shows, this term actually gives the secular equation when $\bs{h}=\bs{0}$ and $z\not=0$.  

Let $\mathcal{C}_n$ denote a circle of some \emph{small}  but fixed radius $\rho_n$ centered at $n\,\uppi$. We have:
\begin{align}\label{eqn:estimate}
	\Big|\frac{s_k(\bs{h})}{z^k}\Big|\leqslant\binom{N}{k}\bigg(\frac{h}{|z|}\bigg)^k,
\end{align}
where $h:=\max\big\{\,|h_1|,\ldots,|h_N|\,\big\}$. Thus, for any $z$ chosen inside $\mathcal{C}_n$ with large $n$, the expression on the left in  Equation~\eqref{eqn:estimate} becomes small. On the other hand, $f_k$'s are all uniformly bounded inside every $\mathcal{C}_n$. To see this, note that  all $f_k$'s are uniformly bounded inside, say, $\mathcal{C}_1$ and they are periodic with period $2\uppi$. Hence,  the second term on the right of Equation~\eqref{eqn:perturbed} becomes negligible with respect to the first term, for all $z$ inside $\mathcal{C}_n$ with large $n$. Since the second term is continuous inside any of the circles $\mathcal{C}_n$,   
 one then expects the secular equation, when $\bs{h}\not=\bs{0}$, has a root that is arbitrary close to $n\,\uppi$, if $n$ is sufficiently large.
 
\subsection{Structure of the  Coefficients $a_j(\bs{h})$}\label{subsec:coeff} 
 
In this subsection we deal with the practical problem of determining the coefficients $a_j(\bs{h})$, which appear in the expansion of the correction term $\mathit{\Delta}_n(\bs{h})$ in Equation~\eqref{eqn:delt_n}. We remind that all even coefficients in the expansion are zero and we  employ a recursive calculation to determine the general structure of odd  coefficients.

We Taylor expand, around $z=0$,  the sine and cosine functions appearing in $D_{\bs{h}}(z)$ given by Equation~\eqref{eqn:secularequation}  and plug $\widetilde{z}_n(\bs{h})=n\,\uppi+\mathit{\Delta}_n(\bs{h})$, with $\mathit{\Delta}_n(\bs{h})$ given by Equation~\eqref{eqn:delt_n}, into this Taylor's expansion and then collect the resultant series in terms of different powers of $n$. The powers of $n$ that appear in this expansion are:
\begin{align*}
	n^{N-1},n^{N-3}, n^{N-5},n^{N-7},\ldots=\{n^{N-2i+1}\}_{i\in\mathbb{N}=\{0,1,2,3,\ldots\}}
\end{align*}
For $\widetilde{z}_n(\bs{h})$ to be a root of the secular equation, all the coefficients of the powers of $n$ introduced above, must vanish. 
 Equating all coefficients to zero and solving the equations recursively, we get the following result for coefficients $a_i(\bs{h})$ with odd indices:
\begin{align}\label{eqn:acoefficents}
	a_{2i-1}(\bs{h})=s_1(\bs{h})\,F_i\big(s_1(\bs{h}),\ldots,s_{i-1}(\bs{h})\big)+\frac{i}{\uppi^{2i-1}N^i}\,\big(s_1(\bs{h})\big)^{i-1}s_i(\bs{h}),\qquad i=2,3,4,\ldots,
\end{align}
where $F_i$'s are some polynomial of degree $2i$ in $i-1$ variables. The coefficients of $F_i$'s are real-valued rational functions of $N$. For $i=1$, the corresponding function $F_1$ is zero. More explicitly, for instance:
\begin{align}\label{eqn:a1}
	a_1(\bs{h})&=\frac{1}{\uppi N}\,s_1(\bs{h}),
\end{align}
and
\begin{align*}
	a_3(\bs{h})&=s_1(\bs{h})F_1\big((s_1(\bs{h})\big)+\frac{2}{\uppi^3 N^2}\,s_1(\bs{h})s_2(\bs{h}),\\
	a_5(\bs{h})&=s_1(\bs{h})F_2\big(s_1(\bs{h}),s_2(\bs{h})\big)+\frac{3}{\uppi^5N^3}\,\big(s_1(\bs{h})\big)^2s_3(\bs{h}),
\end{align*}
where
\begin{align*}
	F_1(x_1)&=-\frac{3N-2}{3\uppi^3N^3}x_1^2-\frac{1}{\uppi^3N^2}x_1,\\
	F_2(x_1,x_2)&=\frac{10N^2-15N+6}{5\uppi^5N^5}x_1^4+\frac{12N-8}{3\uppi^5N^4}x_1^3-\frac{7N-6}{\uppi^5N^4}x_1^2x_2+\frac{2}{\uppi^5N^3}x_1^2-\frac{8}{\uppi^5N^3}x_1x_2+\frac{4}{\uppi^5N^3}x_2^2.
\end{align*}
A closer look at Equations~\eqref{eqn:acoefficents} reveals that $s_i(\bs{h})$ appears, for the first time, in  the equation for $a_{2i-1}(\bs{h})$ and, moreover, it appears there with power equal to one. This we use to prove the main result of this section. 
\subsection{Necessary Condition, The Proof}
Finally, everything is in order and we finilize the proof for the necessary condition of the main theorem.  In fact, here we  assume that the eigenvalues of $L_{\bs{h}}$ possess reflection symmetry with respect to the real axis and we prove that the $N$-tuple $\bs{h}=(h_1,\ldots,h_N)$, consisting of Robin parameters, is invariant under conjugation.

Since the eigenvalues of $L_{\bs{h}}$ possesses reflection symmetry with respect to the real axis, then $D_{\bs{h}}(z)=0$ if and only if $D_{\bs{h}}(\overline{z})=0$ and, by Equation~\eqref{eqn:Dconjugate}, this is equivalent to $D_{\overline{\bs{h}}}(z)=0$. Therefore, $\widetilde{z}_n(\bs{h})$ introduced in Equation~\eqref{eqn:z_n}, is a root of $D_{\overline{\bs{h}}}(z)$ as well. If one uses the same procedure that gave rise to Equation~\eqref{eqn:acoefficents}, but this time for $D_{\overline{\bs{h}}}(z)$, one gets:
\begin{align}\label{eqn:a'1}
	a_1(\bs{h})&=\frac{1}{\uppi N}\,\overline{s_1(\bs{h})},
\end{align}
and
\begin{align}\label{eqn:a'coefficents}
	a_{2i-1}(\bs{h})=\overline{s_1(\bs{h})}\,F_i\big(\overline{s_1(\bs{h})},\ldots,\overline{s_{i-1}(\bs{h})}\big)+\frac{i}{\uppi^{2i-1}N^i}\,\big(\overline{s_1(\bs{h})}\big)^{i-1}\overline{s_i(\bs{h})},\qquad i=2,3,4,\ldots.
\end{align}
Here we used $s_m(\overline{\bs{h}})=\overline{s_m(\bs{h})}$ for $m=1,\ldots,N$.

Comparing Equations~\eqref{eqn:a1} and \eqref{eqn:a'1}, one finds that $s_1(\bs{h})$ is real. Now assume that $s_1(\bs{h})$ till $s_{k-1}(\bs{h})$, $1\leqslant k\leqslant N$, are real. From Equation~\eqref{eqn:a'coefficents}, we have:
\begin{align*}
	a_{2k-1}(\bs{h})=s_1(\bs{h})\,F_k\big(s_1(\bs{h}),\ldots,s_{k-1}(\bs{h})\big)+\frac{k}{\uppi^{2k-1}N^k}\,\big(s_1(\bs{h})\big)^{k-1}\overline{s_k(\bs{h})},
\end{align*}
and from Equation~\eqref{eqn:acoefficents}, we have
\begin{align*}
	a_{2k-1}(\bs{h})=s_1(\bs{h})\,F_k\big(s_1(\bs{h}),\ldots,s_{k-1}(\bs{h})\big)+\frac{k}{\uppi^{2k-1}N^k}\,\big(s_1(\bs{h})\big)^{k-1}s_k(\bs{h}).
\end{align*}
Thus $s_k(\bs{h})$ is also real. Consequently, all polynomials $s_1(\bs{h})$ till $s_N(\bs{h})$ are real and by Lemma~\ref{lem:symmetricpolynomials}, $\bs{h}$ is invariant under conjugation.

\section{Proof of the Main Theorem \ref{thm:maintheorm} and perspectives}

We need to prove that existence of a certain $ \mathcal{PT}$-symmetry is necessary and sufficient for the spectrum of the operator $ L_{\bs{h}} $
to possess reflection symmetry with respect to the real axis. 
Sufficiency (it is a general property of $ \mathcal{PT}$-symmetric operators) is already proven -- it is given by Proposition \ref{prop:symmetry}. Hence it remains to prove the necessity.
Theorem \ref{thm:2} implies that the spectrum possesses reflection symmetry only if the set of Robin parameters $ \{ h_i \} $ is
invariant under complex conjugation. Then  Lemma \ref{lma:PTsymmetric} implies that there exists an automorphism $ \mathcal P$ of $ \Gamma $
(given by a product of edge permutations) such that the operator $ L_{\bs{h}} $ is $ \mathcal{PT}$-symmetric. It follows that every operator $ L_{\bs{h}} $
with reflection-symmetric spectrum possesses $ \mathcal{PT}$-symmetry with a certain automorphism $ \mathcal P $
given as a product of permutations. The Main Theorem is proven.

We have already mentioned in the Introduction that we doubt that the statement that reflection symmetry of
the spectrum implies $ \mathcal{PT}$-symmetry of the quantum graph operator $ L $ with respect to a certain automorphism $ \mathcal P $
of the underlying metric graph holds in full generality. This would imply in particular that all quantum graphs with reflection symmetric
spectrum necessarily are defined on symmetric metric graphs. It might happen that this statement holds if one restricts consideration to
vertex conditions given only via complex delta interactions. We are planning to return back to this question in one of our future publications. 

Note that in the proof of Theorem \ref{thm:2} and therefore in the proof of the Main Theorem we did not use that the multiplicities of
the conjugated eigenvalues coincide. The reason is that only a small part of the eigenvalues was used. One may show that the
eigenvalues of $ L_{\bs{h}} $ are situated close to the eigenvalues of the standard Laplace operator $ L_{\bs{0}}.$ More precisely,
there is exactly one eigenvalue close to $ \uppi n $ and $ N-1 $ eigenvalues close to $ \uppi (n+1/2). $ In our proof we used
just the series of eigenvalues close to $ \uppi n $ and such a series is invariant under conjugation only if it is real. Such eigenvalues cannot be
degenerate. Generalising our results for
arbitrary graphs would require to take into account that the multiplicities in the conjugated pairs of eigenvalues coincide.

Researchers working with $\mathcal{PT}$-symmetric operators are often interested when the operator is not self-adjoint, but the
spectrum is real anyway. One may study this question for quantum graphs getting explicit examples.

\section*{Acknowledgements}
	B.M.G. would like to thank E. Ardonne, I. Mahyaeh,  and F. \v{S}tampach for discussions. He was sponsored, in part, by the Swedish Research Council. P.K. was partially supported by the Center for
  Interdisciplinary Research (ZiF) in Bielefeld in the framework of
  the cooperation group on ``Discrete and continuous models in the
  theory of networks'' and by the Swedish Research Council grant D0497301.\nopagebreak[4]

\appendix
\section{Proofs of Propositions 1 and 2}\label{app:proposition}
In this appendix we prove Propositions \ref{prop:symmetry} and \ref{prop:symmetricpolynomial}.

\begin{proof} (Proposition \ref{prop:symmetry}) Since $\mathcal{AT}$ is a symmetry of $L$, we have:
	\begin{align}\label{eqn:PTsymmetric}
	L\mathcal{AT}=\mathcal{AT}L.
\end{align}
On the other hand, since $\lambda$ is an eigenvalue of $L$ with degeneracy $d$, there exist $d$ linearly independent functions $\bs{u}^1$, $\bs{u}^2$, \dots, and $\bs{u}^d$ in $\dom(L)$ such that
	\begin{align*}
		L\bs{u}^i=\lambda\,\bs{u}^i,\qquad(i=1,\ldots,d).
	\end{align*}
	Acting both sides of the equation above by $\mathcal{AT}$ and considering that $\mathcal{T}$ is anti-linear and $\mathcal{A}$ is linear, one gets:
	\begin{align*}
		\mathcal{A}\mathcal{T}L\bs{u}_i=\overline{\lambda}\,\mathcal{AT}\bs{u}_i,\qquad(i=1,\ldots,d).
	\end{align*}
	Using Equations~\eqref{eqn:PTsymmetric} and \eqref{eqn:Treversal} in the equations above, one gets:
	\begin{align*}
		L\big(\mathcal{A}\overline{\bs{u}}^i\big)=\overline{\lambda}\,\big(\mathcal{A}\overline{\bs{u}}^i\big),\qquad(i=1,\ldots,d).
	\end{align*}
	Now we show that the functions $\mathcal{A}\overline{\bs{u}}^1$, $\mathcal{A}\overline{\bs{u}}^2$, \dots, and $\mathcal{A}\overline{\bs{u}}^d$ are linearly independent. Let $c_1$, $c_2$, \dots, and $c_N$ be complex numbers such that
	\begin{align*}
		c_1\,(\mathcal{A}\overline{\bs{u}}^1)+c_2\,(\mathcal{A}\overline{\bs{u}}^2)+\cdots+c_N\,(\mathcal{A}\overline{\bs{u}}^N)=0.
	\end{align*}
	Then
	\begin{align*}
		\mathcal{A}(c_1\,\overline{\bs{u}}^1+c_1\,\overline{\bs{u}}^2+\cdots+c_1\,\overline{\bs{u}}^N)=0,
	\end{align*}
	and, since $\mathcal{A}$ is invertible, we have:
	\begin{align*}
		c_1\,\overline{\bs{u}}^1+c_1\,\overline{\bs{u}}^2+\cdots+c_1\,\overline{\bs{u}}^N=0,
	\end{align*}
	or
	\begin{align*}
		\overline{c}_1\,\bs{u}^1+\overline{c}_1\,\bs{u}^2+\cdots+\overline{c}_1\,\bs{u}^N=0.
	\end{align*}
	Therefore, all $c_i$'s vanish, since $\bs{u}^i$'s are linearly independent.
\end{proof}

\begin{proof} (Proposition \ref{prop:symmetricpolynomial}) To prove Equation~\eqref{eqn:reccursive}, because of symmetry, one can assume that $i=N$ without loss of generality.  

In what follows, $\sum^{(1)}$ indicates a sum over all subscripts such that the lowest possible value for subscripts is $1$ and the highest possible value for subscripts is $N$, and $\sum^{(2)}$ indicates the same except that the highest possible value for subscripts is $N-1$ instead.
\begin{align*}
	\text{RHS}&=\mathlarger{\mathlarger{\textstyle{\sum}}}^{(1)}c_{i_1}\!\!\cdots c_{i_{j+1}}-\mathlarger{\mathlarger{\textstyle{\sum}}}^{(2)}c_{i_1}\!\!\cdots c_{i_{j+1}}\\&=\Big(c_N\mathlarger{\mathlarger{\textstyle{\sum}}}^{(2)}c_{i_1}\!\!\cdots c_{i_{j}}+\mathlarger{\mathlarger{\textstyle{\sum}}}^{(2)}c_{i_1}\!\!\cdots c_{i_{j+1}}\Big)-\mathlarger{\mathlarger{\textstyle{\sum}}}^{(2)}c_{i_1}\!\!\cdots c_{i_{j+1}}\\&=c_N\,s_j(\bs{c}_N)\\&=\text{LHS}.
\end{align*}
Equation~\eqref{eqn:reccursive} then follows from Equation~\eqref{eqn:023}. 
\end{proof}

\end{document}